\documentclass[
onecolumn,
showpacs,preprintnumbers,amsmath,amssymb,aip,jmp,superscriptaddress,footinbib
]{revtex4-1}

\usepackage{amsmath, amsthm, amssymb, amsfonts, amscd, graphicx, verbatim}
\usepackage{subfigure}
\newtheorem{theorem}{Theorem}[section]

\newtheorem{proposition}[theorem]{Proposition}
\newtheorem{corollary}[theorem]{Corollary}

\newtheorem{example}[theorem]{Example}

\def\IR{{\mathbb R}}
\def\IC{{\mathbb C}}

\def\cS{{\cal S}}

\def\tr{{\rm tr}}
\def\rank{{\rm rank}\,}

\newcommand{\bra}[1]{\mbox{$\left\langle #1 \right|$}}
\newcommand{\ket}[1]{\mbox{$\left| #1 \right\rangle$}}

\begin{document}
\openup 1\jot
\title{
Entanglement transformation between sets of bipartite pure quantum states using local operations
}

\author{H.~F. Chau}
\affiliation{Department of Physics and Center of Computational and Theoretical
 Physics, University of Hong Kong, Pokfulam Road, Hong Kong}
\author{Chi-Hang Fred Fung}
\affiliation{Department of Physics and Center of Computational and Theoretical
 Physics, University of Hong Kong, Pokfulam Road, Hong Kong}
\author{Chi-Kwong Li}
\affiliation{Department of Mathematics, College of William \& Mary, Williamsburg, Virginia 23187-8795, USA}
\author{Edward Poon}
\affiliation{Department of Mathematics, Embry-Riddle Aeronautical University, 3700 Willow Creek Road, Prescott, Arizona 86301, USA}
\author{Nung-Sing Sze}
\affiliation{Department of Applied Mathematics, The Hong Kong Polytechnic University, Hung Hom, Hong Kong}


\begin{abstract}
Alice and Bob are given an unknown initial state chosen from a set of pure quantum states.
Their task is to transform the initial state to a corresponding final pure state using local operations only.
We prove necessary and sufficient conditions on the existence of such a transformation.
We also provide efficient algorithms that can quickly rule out the possibility of transforming a set of initial states to a set of final states.
\end{abstract}

\date{\today}


\maketitle


\section{Introduction}

Entanglement is a useful resource in quantum information processing.
Entanglement enables two parties to share non-classical correlation and
is the basis for unconditional security in
quantum key distribution~\cite{PhysRevLett.67.661,PhysRevLett.68.557,LoChauQKD_99},
exponential speedup in
quantum computing~\cite{Deutsch1985,Shor1994,DiVincenzo1995,Lloyd1996},
and
error tolerance in computing with
quantum error correction codes~\cite{Shor1996,Cirac1996,PhysRevA.57.127}.
In this paper, we study a simple problem of transformation of entangled states,
which is perhaps the simplest non-trivial case of studying the state transformation problem in the
multi-partite setting.
Alice and Bob are initially given an unknown bipartite pure state
from the set
$\{\ket{x_i}_{AB} \in \mathcal{H}_A\otimes \mathcal{H}_B, i=1,\ldots,N\}$, where $\mathcal{H}_A$ ($\mathcal{H}_B$) denotes the Hilbert space of Alice's (Bob's) quantum system with dimension $m$ ($n$).
Their task is to transform the initial state to a corresponding state
in the set 
$\{\ket{y_i}_{AB} \in \mathcal{H}_A\otimes \mathcal{H}_B, i=1,\ldots,N\}$ 
by performing local quantum operations only without any communication with each other.
Any quantum operation on a density matrix $\rho$ can be represented by a trace preserving completely positive (TPCP) map which can be expressed in the operator-sum form $\mathcal{F}(\rho)=\sum_i F_i \rho F_i^\dag$ where $\sum_i F_i^\dagger F_i=I$.
We prove necessary and sufficient conditions for the existence of a transformation
$(\mathcal{F} \otimes \mathcal{G}) (\ket{x_i}_{AB}\bra{x_i})=\ket{y_i}_{AB}\bra{y_i}$ for all $i=1,\ldots,N$,
where Alice (Bob) performs physical process $\mathcal{F}$ ($\mathcal{G}$) on
 $\mathcal{H}_A$ ($\mathcal{H}_B$).
%
%
In other words,
we are interested in the existence of
a local operation (LO) transformation $T$ of the form
\begin{equation}\label{form}
T(\rho_{AB}
) = \sum_{1\le i\le p, 1 \le j \le q} (F_i \otimes G_j)
\rho_{AB}
(F_i^\dag \otimes G_j^\dag)
\end{equation}
such that
$$T (\ket{x_i}_{AB}\bra{x_i})=\ket{y_i}_{AB}\bra{y_i} \text{ for all } i=1,\ldots,N,$$
where
$F_1, \dots, F_p \in M_m$ and
$G_1, \dots, G_q \in M_n$ are Kraus operators satisfying
$\sum F_i^\dag F_i = I$ and $\sum G_j^\dag G_j = I$.
%
Here, $p$ and $q$ are the dimensions of the ancillas for the channels $\mathcal{F}$ and $\mathcal{G}$.
Figure~\ref{fig:transform_state} depicts the bipartite state transformation problem that we consider in this paper.

The state transformation problem has been studied since the early 1980s.
Alberti and Uhlmann~\cite{Alberti1980163} proved necessary and sufficient conditions for the existence of a physical process that transforms two qubit (mixed) states to two other qubit (mixed) states.
Subsequent works found necessary and sufficient conditions for transformations between two sets of pure states without any restriction on the number of states~\cite{Uhlmann1985,Chefles200014,PhysRevA.65.052314,Chefles_Jozsa_Winter_2004}.
Note that these previous results assume that the physical process acts on the entire Hilbert space the states live in (i.e., they do not impose a bipartite structure on the transformation as we do here).
The pure-state single-party result~\cite{Uhlmann1985,Chefles200014,PhysRevA.65.052314,Chefles_Jozsa_Winter_2004} states that a TPCP map $T$ exists such that
$T(\ket{x_i}_{AB}\bra{x_i})=\ket{y_i}_{AB}\bra{y_i}$
for $i = 1, \dots, N$ if and only if
there is a correlation matrix (positive semidefinite matrix with all diagonal entries equal to one)
$M \in M_N$ such that the Gram matrix
$(\langle x_i \ket{x_j})$
equals
the Schur (entrywise) product $M\circ (\langle y_i \ket{y_j})$.

Along a different line of research, transformation from one bipartite entangled pure state to another using local operations and classical communications (LOCC) 
has been studied for a long time.
In particular, 
Bennett {\it et al.}~\cite{Bennett1996} as well as Lo and Popescu~\cite{Lo:2001:concentration} studied the transformation of a bipartite state to a maximally entangled state using LOCC.
In fact,
Lo and Popescu~\cite{Lo:2001:concentration} showed that any general  transformation between bipartite pure states using LOCC can be performed with one-way classical communications only.
Along a similar line,
Nielsen~\cite{PhysRevLett.83.436} studied the transformation problem between two general bipartite states and 
proved
that $\ket{x}_{AB}$ transforms to $\ket{y}_{AB}$ using LOCC if and only if the eigenvalues of $\operatorname{tr}_B(\ket{x}_{AB}\bra{x})$ are majorized by those of $\operatorname{tr}_B(\ket{y}_{AB}\bra{y})$.
Jonathan and Plenio~\cite{PhysRevLett.83.1455} extended this result to the case where the transformation is allowed to output one of many final states with a certain probability.
Again, a majorization condition has been proven for this case.
He and Bergou~\cite{He2008} further proposed a scheme to transform one bipartite state to another probabilistically.
Gheorghiu and Griffiths~\cite{Gheorghiu2008} considered transformation using separable operations (of which LOCC is a subset) and proved the same majorization condition for transformability as that in Ref.~\citenum{PhysRevLett.83.1455}.
Note that these results concern the transformation of only one bipartite state to another bipartite state.

In this paper, we study the combination of the previous two research directions by considering the transformation of sets of bipartite states.
Our result shares some characteristic with the single-state bipartite result~\cite{PhysRevLett.83.436} in that singular (or eigen) values of the initial and final states play an important role.
In particular, we can immediately rule out the existence of the desired transformation by looking at the singular (or eigen) values in both cases.
On the other hand, while the eigenvalues alone are sufficient to determine whether a LOCC map exists to transform one bipartite state to another~\cite{PhysRevLett.83.436},
the eigenvalues alone are not sufficient when we consider the LO transformation of a set of bipartite states to another.
In summary, we consider in this paper the transformability problem in a new context that was only partially considered by previous results, and as discussed our necessary and sufficient conditions are different from those in the other contexts studied previously.

We remark that, along a different direction, some of us have studied the problem of transforming between two sets of {\it mixed} states (with a single-party TPCP map, instead of a two-party map using LO or LOCC), and obtained necessary and sufficient conditions for the existence of such a transformation~\cite{Huang2012}.

%
%
%
%
%

\begin{figure}
\begin{center}
\subfigure[]{
\includegraphics[width=.5\columnwidth]{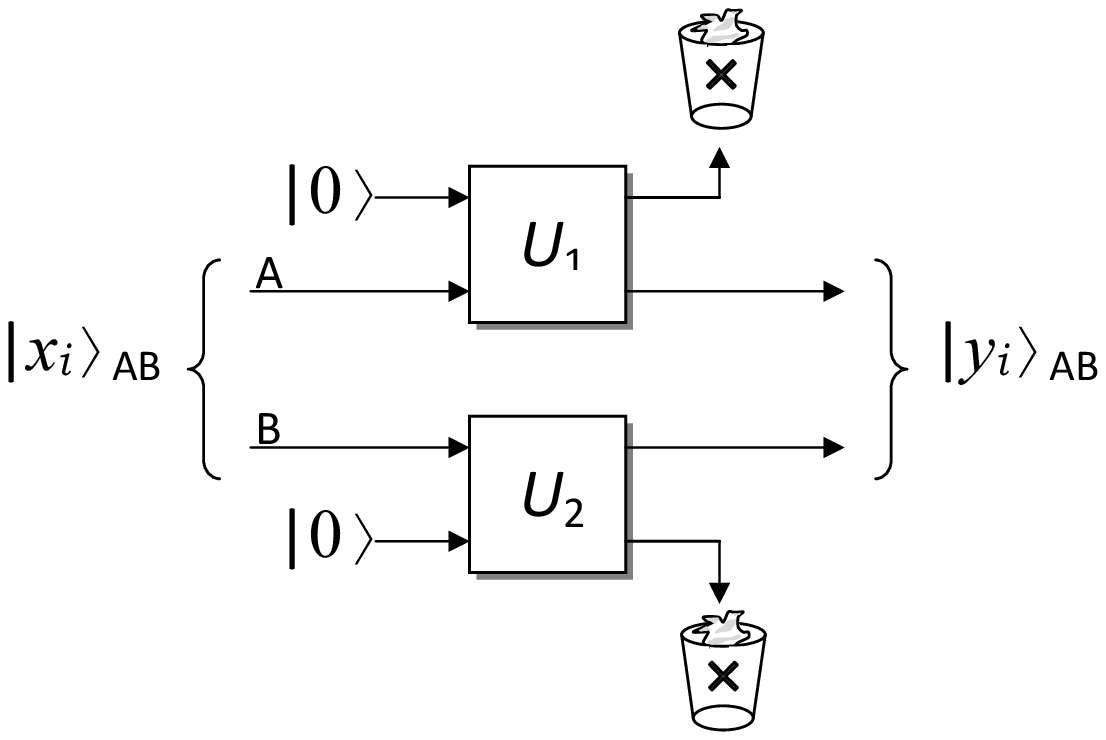}\label{fig:transform_state}}
\\
\subfigure[]{
\includegraphics[width=1\columnwidth]{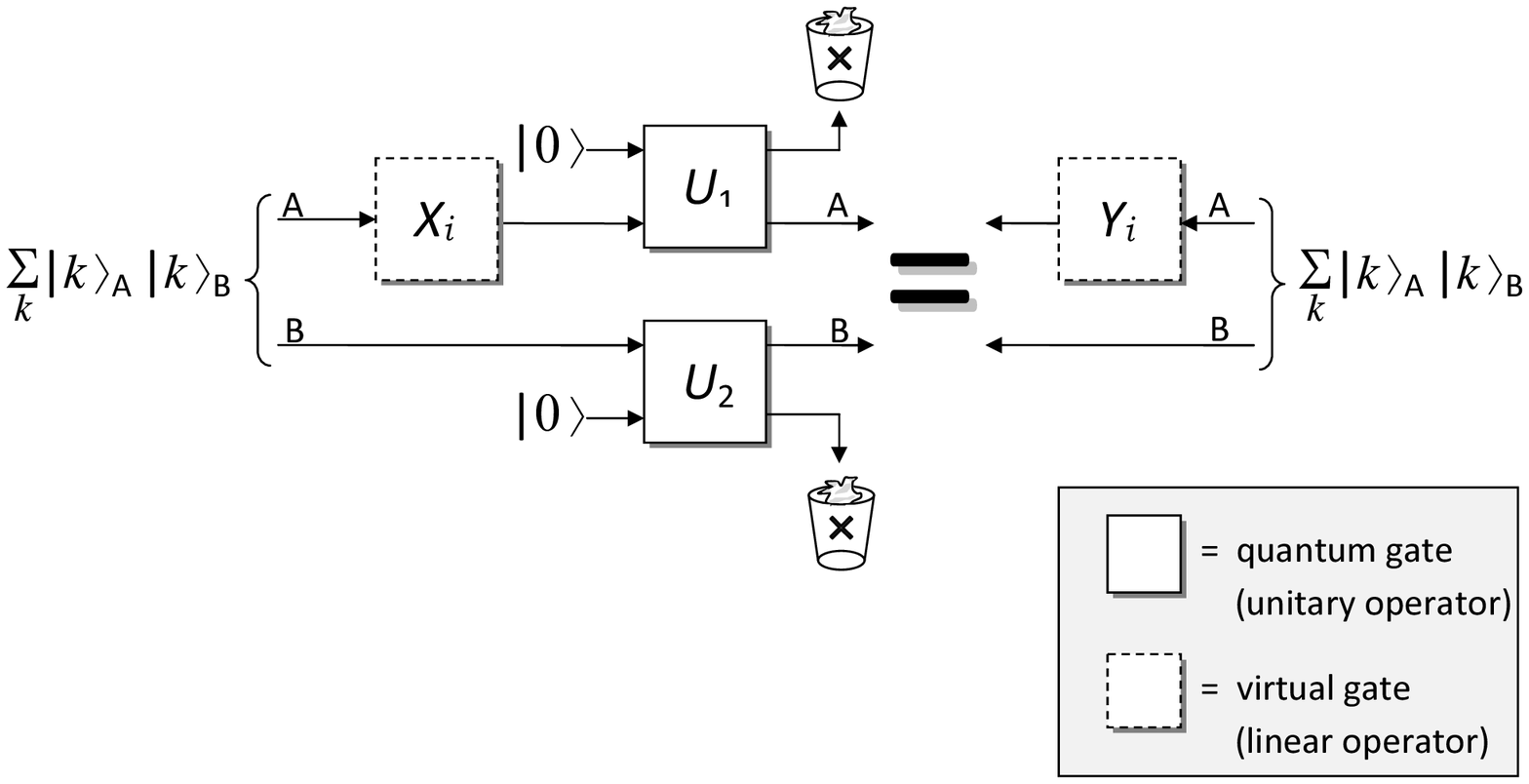}\label{fig:transform_linear_operator}}
\caption{\label{fig:DPA}
The bipartite state transformation problem.
(a) State representation of the problem: The input states $\ket{x_i}_{AB}$ are transformed to the corresponding output states $\ket{y_i}_{AB}$ for all $i=1,\dots,N$ through a LO transformation $U_1 \otimes U_2$.
The ancillas are prepared in the standard states and are discarded after the transformation.
(b) Gate representation of the problem: The input states $\ket{x_i}_{AB}$ (output states $\ket{y_i}_{AB}$) are mapped to linear operators $X_i$ ($Y_i$).
The LO transformation $U_1 \otimes U_2$ maps the input ``gate'' $X_i$ to the output ``gate'' $Y_i$ for all $i$.
This gate representation is useful as it can succinctly capture our necessary and sufficient conditions for the existence of the LO transformation.
}
\end{center}
\end{figure}

\subsection{State representation}

We always work in some fixed orthonormal bases
$\{ \ket{i}_A : i=1,\dots,m\}$ and $\{ \ket{j}_B: j=1,\dots,n \}$ on Alice's and Bob's systems respectively throughout the paper.
Thus,
a joint state of Alice and Bob
$\ket{x}_{AB}=\sum_{i=1,\dots m; j=1,\dots n} \lambda_{ij} \ket{i}_A \ket{j}_B$
(where  $\sum_{ij} |\lambda_{ij}|^2=1$)
can be expressed
as a column vector
\begin{equation}
\label{eqn-vector-form1}
x=
\begin{bmatrix}
\lambda_{11}\\
\lambda_{12}\\
\vdots\\
\lambda_{mn}
\end{bmatrix}
\in \IC^{mn} .
\end{equation}

For any state $\ket{x}_{AB}=\sum_{i=1,\dots,m; j=1,\dots,n} \lambda_{ij} \ket{i}_A \ket{j}_B$, we define an associated linear operator
$X=\sum_{i=1,\dots,m; j=1,\dots,n} \lambda_{ij} \ket{i} \bra{j}$, which
is conveniently represented as an $m \times n$ matrix with the $(i,j)$ element being $\lambda_{ij}$:
\begin{equation}
\label{eqn-matrix-form1}
X=
\begin{bmatrix}
\lambda_{11} & \lambda_{12} & \dots & \lambda_{1n} \\
\lambda_{21} & \ddots & \dots & \lambda_{2n} \\
\vdots & & & \vdots \\
\lambda_{m1} & & \dots & \lambda_{mn}
\end{bmatrix}.
\end{equation}
Note that here we are using the well known correspondence between vectors and matrices: $x=\operatorname{vec}(X^t)$ (see, e.g., Ref.~\citenum{Horn1994}).
The $\operatorname{vec}$ operator has a property that relates multiplication and tensor product: $\operatorname{vec}(AXB)=(B^t \otimes A) \operatorname{vec}(X)$.
This property will be useful in Theorem~\ref{2.1} later.
In order to simplify notations in the rest of this paper, whenever we refer to the correspondence between a vector and a matrix, they are related by Eqs.~\eqref{eqn-vector-form1} and
\eqref{eqn-matrix-form1}, with a lower-case (upper-case) letter representing the vector (matrix).

We slightly abuse the notation by calling $X$ a virtual quantum gate that changes the space of the state it acts on from $n$-dimensional to $m$-dimensional.
It turns out that our main results and proofs can be conveniently expressed in terms of $X_i$
and
$Y_i$
in addition to the original states $\ket{x_i}_{AB}$ and $\ket{y_i}_{AB}$, where $i=1,\dots,N$.
First note that $(X \otimes I_B) \sum_{j=1,\dots,n} \ket{j}_{A} \ket{j}_{B}=\ket{x}_{AB}$ for any
$\ket{x}_{AB}$.
This means that the same input state $\sum_{j} \ket{j} \ket{j}$ (which is an unnormalized maximally entangled state) can always be transformed to any given state using its associated linear operator.
Then the state transformation problem can be visualized as finding a quantum circuit ($U_1$ and $U_2$) that preserves the state $\sum_{j} \ket{j} \ket{j}$ processed by virtual gates $X_i$ and $Y_i$ where $i=1,\dots,N$ (see Fig.~\ref{fig:transform_linear_operator}).
Even though $X_i$ and $Y_i$ are not ordinary quantum gates because they are not unitary, they still transform states as linear operators and so we can still capture them in a quantum-circuit-like diagram.
As we show below, the gate representation with $X_i$ and $Y_i$ is useful as it can capture our main results succinctly.

Using these notations, 
we have
$\tr_B \ket{x}_{AB}\bra{x}=XX^\dag$.

\subsection{Notations}

In what follows, we adopt the mathematicians' notation that the superscript $^*$ denotes conjugate transpose and $^t$ denotes transpose.
Thus, $x x^*$ denotes $\ket{x}\bra{x}$.

\section{Main results}

\begin{theorem} \label{2.1}
Given pure states
$x_1x_1^*, \dots, x_kx_k^*, y_1y_1^*,\dots, y_ky_k^* \in H_{mn}$, let
$X_i$ and $Y_i$ be the $m\times n$ matrix forms of $x_i$ and $y_i$ for $i = 1, \dots, k$.
Let $E_{11}$ be the $p\times q$ matrix with 1 at the $(1,1)$ position and 0 elsewhere.
The following conditions are equivalent.

{\rm (a)} There is a TPCP map $T$ of the form (\ref{form}) on tensor states such that
$T(x_ix_i^*) = y_iy_i^*$ for $i = 1, \dots, k$.

{\rm (b)} There is a unitary $U = U_1\otimes U_2
\in M_{mp} \otimes M_{nq}$ such that
$$U\begin{bmatrix} \tilde x_1 & \cdots \tilde x_k \end{bmatrix} =
\begin{bmatrix} \tilde{y}_1 & \cdots & \tilde{y}_k \end{bmatrix},
$$
where $\tilde x_i$ is the vector form of $E_{11} \otimes X_i$,
$\tilde{y}_i$ is the vector form of $R_i \otimes Y_i$ and $R_1, \dots, R_k \in M_{pq}$ have Frobenius norm one.

{\rm (c)} There are unitary $U\in M_{mp}$, $V \in M_{nq}$,
and matrices $R_1, \dots, R_k \in M_{pq}$
such that
$$\tr R_iR_i^* = 1 \quad \hbox{ and } \quad
U (E_{11} \otimes X_i)V = R_i\otimes Y_i,  \qquad  i = 1, \dots, k.$$
\end{theorem}

\begin{proof}
Note that, if we express $z \in \IC^{mn}$ as an $m \times n$ matrix $Z$, then for any $F \in M_m$ and $G \in M_n$,
the matrix form of $(F \otimes G)z$ is $FZG^t$.
Thus, there exist $F_1, \dots, F_p \in M_m$ and $G_1, \dots, G_q \in M_n$ with
$\sum F_i^* F_i = I_m$ and $\sum G_j^*G_j = I_n$
such that
\begin{equation*}
\sum_{1\le i\le p, 1 \le j \le q} (F_i \otimes G_j) x_l x_l^* (F_i \otimes G_j)^* = y_l y_l^* \quad \forall \: l=1, \dots, k
\end{equation*}
if and only if
\begin{align*}
&(F_i \otimes G_j) x_l = r_{ij}^l y_l \; \text{ for some scalars } r_{ij}^l
\text{ satisfying } \sum_{1 \leq i \leq p, 1 \leq j \leq q} |r_{ij}^l|^2 = 1 \\
& (\text{and} \quad \sum F_i^* F_i = I_m, \sum G_j^* G_j = I_n) \\
\iff & F_i X_l G_j^t = r_{ij}^l Y_l \text{ with above constraints on } r_{ij}^l, F_i, G_j \\
\iff & \begin{bmatrix} F_1 & * & * & * \\ F_2 & * &* &* \\ \vdots &* &* &* \\ F_p &* &* &* \end{bmatrix} \begin{bmatrix} X_l & 0 & \dots & 0 \\ 0 & 0 & & \\ \vdots & & \ddots & \\ 0 & \dots & \dots & 0 \end{bmatrix} \begin{bmatrix} G_1^t & G_2^t & \dots & G_q^t \\ *&*&*&* \\ *&*&*&* \\ *&*&*&* \end{bmatrix} = \begin{bmatrix} r_{11}^l Y_l & r_{12}^l Y_l & \dots & r_{1q}^l Y_l \\ r_{21}^1 Y_l & r_{22}^l Y_l & & \\ \vdots & & \ddots & \\r_{p1}^l Y_l & \dots & \dots & r_{pq}^l Y_l \end{bmatrix} \\
&\text{with previous constraints}\\
\iff & \exists \text{ unitary } U \in M_{pm}, V \in M_{qn}
\text{ such that } U(E_{11} \otimes X_l)V = R_l \otimes Y_l \quad \forall \: l=1, \dots, k \\
&\text{where } R_l \in M_{pq} \text{ has Frobenius norm one} .
\end{align*}
This shows that (a) and (c) are equivalent.  By taking the vector form of the last equation we obtain equivalence with (b).
\end{proof}

\medskip\noindent
\begin{example}
{\rm

Clearly, the existence of unitaries $U \in M_m$ and $V \in M_n$ (in other words, $p=q=1$, the case without ancilla) such that $UX_i V = Y_i$ for all $i=1, \dots, k$ is a sufficient condition for a TPCP map of the desired form.  It is not, however, necessary.
Fix $s_1 > s_2 > 0$.  Let $X_1 = Y_1 = \begin{bmatrix} s_1 & 0 & 0 & 0 \\ 0 & s_2 & 0 & 0 \end{bmatrix}$.  Let $X_2 = \begin{bmatrix} I_2 & I_2 \end{bmatrix}$ and $Y_2 = \begin{bmatrix} \frac{1}{c} I_2 & \frac{r}{c} I_2 \end{bmatrix}$, where $r, c \in \IR$ satisfy $1+r^2 = 2c^2$ and $|c|<1$.  (For example, $c=1/\sqrt{2}$ and $r=0$ will suffice.)
\medskip

Let $\gamma = \sqrt{1-c^2}/c$ and define $$U = I_2 \quad \text{and} \quad V = \begin{bmatrix} I_2 & 0 & 0 & 0 \\ 0 & rI_2 & \gamma I_2 & r\gamma I_2 \\ 0 & * & * & * \\ 0 & * & * & * \end{bmatrix} \in M_8,$$
where we can complete $V$ to be a unitary matrix (since the first four rows are orthonormal).  Then $U(E_{11} \otimes X_1) V = E_{11} \otimes Y_1$ and $U(E_{11} \otimes X_2) V = R \otimes Y$ where $R= \begin{bmatrix} c & \gamma c \end{bmatrix}$, so there is a tensor TPCP map that will interpolate. (Note $\tr R^*R =1$.)
\medskip

However, if $U \in M_2$ and $V \in M_4$ are unitary matrices such that $UX_1V=Y_1$, then $U$ must be a diagonal unitary and $V$ must have the form $V = \begin{bmatrix} U^* & 0 \\ 0 & W \end{bmatrix}$ for some unitary $W \in M_2$.  But then $UX_2V = \begin{bmatrix} I_2 & UW \end{bmatrix} \ne Y_2.$

\medskip
To be more specific, in the above example, let  $s_1 = 2s_2 = 2/3$, $c = 1/\sqrt 2$, and $r = 0$.
If $\{e_1, \dots, e_{16}\}$ is the standard basis  for $\IC^{16}$,
then for any choice of unitary $W \in M_{16}$ with the first, third, sixth, and eighth
columns equal to $e_1, e_9, e_6$, and $e_{14}$, respectively, we have
$W[x_1 \, x_2]= [y_1 \, y_2]$, where $x_1 = y_1 = s_1 e_1 + s_2 e_6$,
$x_2 = e_1 + e_3 + e_6+e_8$, $y_2 = (e_1 + e_6)/c$ are the vector form of
the matrices $X_1=Y_1, X_2, Y_2$.
Clearly, there are many choices of $W$ and not every choice will yield a
matrix of the form $U_1\otimes U_2$ with $U_1 \in M_2$ and $U_2 \in M_8$.

}
\end{example}

\medskip
Condition (c) of Theorem \ref{2.1} shows that it is very difficult to have
an LO operation to transform the states.
For example, one needs to have $R_i$ such that
$E_{11} \otimes X_i$ and $R_i\otimes Y_i$ are unitarily equivalent via the same $U, V$ pair.
In particular, $E_{11}\otimes X_i$ and $R_i \otimes Y_i$
have the same singular values. Thus, the rank of  $X_i$ must be a multiple
of that of $Y_i$. For an individual pair of matrices $X$ and $Y$, this condition
is easy to check as shown in the following proposition.

\begin{proposition} \label{2.2}
Let $x,y \in \IC^{mn}$ be unit vectors and let $X,Y \in M_{mn}$ be their corresponding matrix forms.
There is a tensor TPCP map sending $xx^*$ to $yy^*$ if and only if $X$ and $R \otimes Y$ have the same nonzero singular values for some matrix $R$ with Frobenius norm one.

Furthermore,
suppose $X$ has nonzero singular values $\alpha_1 \geq \dots \geq \alpha_p > 0$ and
$Y$ has nonzero singular values $\beta_1 \geq \dots \geq \beta_q > 0$. Then
the existence of $R$ can be determined by the following algorithm:

Set $A_1 = \{\alpha_1, \dots, \alpha_p \}$.
If $r=p/q$ is not an integer, no map exists.  Otherwise, perform the following.
\begin{align*}
\text{For } & i = 1, \dots, r \\
& \text{Set } \gamma_i = (\max A_i) / \beta_1; \\
& \text{If } \cS_i = \{\beta_1 \gamma_i, \beta_2 \gamma_i, \dots, \beta_q \gamma_i \}
\subset A_i \text{ then set } A_{i+1} = A_i \setminus \cS_i; \\
& \text{otherwise stop: no map exists.}
\end{align*}
If this program finishes with $A_{r+1} = \emptyset$, then take $R$ to be
any matrix with nonzero singular values $\gamma_1, \dots, \gamma_r$,  and a map exists.
(Note that $R$ will automatically have Frobenius norm one.)
\end{proposition}

\begin{proof} The first assertion follows from Theorem \ref{2.1}. One readily verifies
the algorithm.
\end{proof}

We illustrate the algorithm in the following.

\medskip\noindent
\begin{example}
{\rm
For simplicity, we ignore normalization in the following.
Suppose $A_1=\{4,2,2,1\}$.
If $\{\beta_i\}=\{2,1\}$, then we can find $\{\gamma_i\}=\{2,1\}$ to produce $A_1$.
If $\{\beta_i\}=\{4,2\}$, then we can find $\{\gamma_i\}=\{1,1/2\}$ to produce $A_1$.
On the other hand, if $\{\beta_i\}=\{2,1,1\}, \{2,1/2\}$, or $\{1,1\}$, then no $\{\gamma_i\}$ exists to produce $A_1$.
}
\end{example}

\medskip
One easily deduces more necessary conditions for the existence of LO map in Theorem \ref{2.1}.

\begin{corollary}
 Use the notation in Theorem~\ref{2.1}. If any one (and hence all) of the
conditions (a) -- (c) hold, then for any $1 \le i \le j \le k$,
$$U (E_{11}E_{11}^* \otimes X_iX_j^*) U^* = R_iR_j^* \otimes Y_i Y_j^* \qquad \hbox{ and } \qquad
V^* (E_{11}^* E_{11} \otimes X_i^*X_j) V = R_i^*R_j \otimes Y_i^* Y_j.$$
In particular,  the  eigenvalues of $X_i X_j^*$ can be obtained from those of
$Y_iY_j^*$ by taking their multiple with
some $\gamma_1^{(ij)}, \dots, \gamma_{\ell_{ij}}^{(ij)} \in \IC$.
\end{corollary}

\medskip
If $X_1, \dots, X_k$ and $Y_1, \dots, Y_k$ are all rank one matrices, we have the following.

\begin{corollary} \label{2.3} Suppose the matrices $X_i, Y_i$
 for $i = 1, \dots, k$ constructed in
Theorem \ref{2.1} have rank one. Let
$X_i = a_i b_i^*, Y_i = c_i d_i^*$
with unit vectors $a_i, c_i \in \IC^m, b_i, d_i \in \IC^n$ for $i = 1, \dots, k$.
Then conditions (a) - (c) in Theorem \ref{2.1} are equivalent to:

\medskip
{\rm (d)} There are unitary matrices $U \in M_{mp}$, $V \in M_{nq}$, and
unit vectors $\zeta_1, \dots, \zeta_k \in \IC^p$, $\eta_1, \dots, \eta_k \in \IC^q$
such that
$$U\begin{bmatrix}a_1 & \cdots & a_k\cr 0& \cdots & 0 \cr\end{bmatrix}
= \begin{bmatrix}\zeta_1 \otimes c_1& \cdots & \zeta_k \otimes c_k\cr\end{bmatrix}
\ \hbox{ and } \
V^*\begin{bmatrix} b_1 & \cdots & b_k\cr 0& \cdots & 0 \cr\end{bmatrix}
= \begin{bmatrix}\eta_1 \otimes d_1& \cdots & \eta_k \otimes d_k\cr\end{bmatrix}.$$
\end{corollary}

The next theorem reduces the general problem to the situation described in  Corollary \ref{2.3}.
The key idea is to use the fact that every $m\times n$ matrix $X$ of rank $r$
admits a Schmidt decomposition  $X = \sum_{j=1}^r s_j(X) a_j b_j^*$, where
$s_1(X) \ge \cdots \ge s_r(X) > 0$, and
$\{a_1, \dots, a_r\} \subseteq \IC^m$
and  $\{b_1, \dots, b_r\} \subseteq \IC^n$
are orthonormal sets.

\begin{theorem} \label{2.4} Use the notation in Theorem \ref{2.1}.
Then any of the conditions (a) - (c)
in Theorem \ref{2.1} holds only if

\medskip
{\rm (e)} for each $i = 1, \dots, k$, the matrix
$Y_i$ has a Schmidt decomposition
$Y_i = \sum_{v=1}^{\rank(Y_i)} s_{v}(Y_i) Y_{iv}$ and there are positive numbers
$\gamma_{i1}, \dots, \gamma_{i\ell_i}$ with $\ell_i = \rank(X_i)/\rank(Y_i)$
such that $X_i$ has a Schmidt decomposition
$X_i = \sum_{u,v} \gamma_{iu} s_v(Y_i) X_{iuv}$ with $X_{iuv} \in M_{m,n}$,
and there is a tensor TPCP map sending $X_{iuv}$ to $Y_{iv}$
for all $i,u,v$.
\end{theorem}

\begin{proof}
We only need to show that (c) implies (e).
If (c) holds, then
$E_{11} \otimes X_{i} = U^*(R_i\otimes Y_i)V^*$ for each $i = 1, \dots, k$.
Taking a Schmidt decomposition for each of $R_i = \sum s_u(R_i) R_{iu}$ and
$Y_i = \sum s_v(Y_i) Y_{iv}$, where $R_{iu}$ and $Y_{iv}$ have rank one and Frobenius norm one,
we see that $E_{11} \otimes X_i$ has a Schmidt decomposition
$$E_{11} \otimes X_i = \sum_{u,v} s_u(R_i)s_{v}(Y_i) \tilde X_{iuv}
\qquad \hbox{ with } \qquad
\tilde X_{iuv} = U^*(R_{iu} \otimes Y_{iv})V^*.$$
We claim that $\tilde X_{iuv} = E_{11} \otimes X_{iuv}$ with $X_{iuv} \in M_{m,n}$.
To see this, note that $\tilde X_{iuv} = \tilde a_{iuv}\tilde b_{iuv}^*$ for some
unit vectors $\tilde a_{iuv} \in \IC^{mp}$ and $\tilde b_{iuv}\in \IC^{nq}$
such that
$$(E_{11}\otimes X_i) \tilde b_{iuv} = s_u(R_i)s_v(Y_i)\tilde a_{iuv} \quad \hbox{ and } \quad
(E_{11}\otimes X_i)^*\tilde a_{iuv} = s_u(R_i)s_v(Y_i) \tilde b_{iuv}.$$
Thus,
$$\tilde a_{iuv} = \begin{bmatrix}a_{iuv} \cr 0\cr\end{bmatrix} \qquad \hbox{ and }
\qquad \tilde b_{iuv} = \begin{bmatrix}b_{iuv} \cr 0\cr\end{bmatrix}$$
with $a_{iuv} \in \IC^m$ and $b_{iuv} \in \IC^n$.
Thus, $\tilde X_{iuv} = E_{11} \otimes X_{iuv}$ with $X_{iuv} = a_{iuv}b_{iuv}^*$.
Consequently, the tensor TPCP map described in (c) will send $X_{iuv}$ to $Y_{iv}$.
\end{proof}

Note that the converse $(e) \implies (c)$ does not hold in general.  However, one can readily verify that the converse holds if one makes the additional stipulation that there exist $R_{iu} \in M_{pq}$ and unitary $U \in M_{mp}$, $V \in M_{nq}$ such that $U(E_{11} \otimes X_{iuv})V = R_{iu} \otimes Y_{iv}$ for all $i,u,v$.

By the result in Ref.~\citenum{Chefles_Jozsa_Winter_2004}, we can extend Theorem \ref{2.1} to the
following.

\begin{corollary} Suppose $A_1, \dots, A_k \in M_{mn}$ are mixed states,
and $y_1 y_1^*, \dots, y_ky_k^*\in M_{mn}$ are pure states such that
$A_i = \sum_{j=1}^{\ell_i} x_{ij}x_{ij}^*$ for $i = 1, \dots, k$, where
$\ell_1, \dots, \ell_k$ are positive integers.
Then there is a TPCP map in the tensor form {\rm (\ref{form})} such that
$T(A_i) = y_iy_i^*$ for all $i = 1, \dots, k$ if and only if
there is a TPCP map of the tensor form {\rm (\ref{form})} such that
$T(x_{ij}x_{ij}^*) = (x_{ij}^*x_{ij})y_iy_i^*$ for all $i = 1, \dots, k$, and
$j = 1, \dots, \ell_i$.
\end{corollary}

\begin{proof} By the result in Ref.~\citenum{Chefles_Jozsa_Winter_2004}, there is a TPCP map
$T$ sending the mixed states $A_1, \dots, A_k$ to $y_1y_1^*, \dots, y_ky_k^*$
if and only if there is a TPCP map sending $x_{ij}x_{ij}^*$ to
$(x_{ij}^*x_{ij})y_iy_i^*$ for all $i = 1, \dots, k$, and
$j = 1, \dots, \ell_i$. Clearly, if $T$ has the tensor form (\ref{form}) and sends
$A_i$ to $y_iy_i^*$, then $y_iy_i^* = T(A_i) = \sum T(x_{ij}x_{ij}^*) \ge T(x_{ij}x_{ij}^*)$,
where $P \ge Q$ means that $P-Q$ is positive semidefinite. Considering the
range space and the trace of  $T(x_{ij}x_{ij}^*)$, we see that
$T(x_{ij}x_{ij}^*) = (x_{ij}^*x_{ij})y_iy_i^*$. The converse is clear.
\end{proof}

We close this section with the following example.

\medskip\noindent
\begin{example}
{\rm
Suppose
\begin{eqnarray}
X_1=
\frac{1}{\sqrt{85}}
\begin{bmatrix}
8 & 0 & 0 & 0\\
0 & 2 & 0 & 0\\
0 & 0 & 4 & 0\\
0 & 0 & 0 & 1
\end{bmatrix},
\:
Y_1=
\frac{4}{\sqrt{20}} a_1 b_1^*+
\frac{2}{\sqrt{20}} a_2 b_2^*,
\end{eqnarray}
\begin{eqnarray}
X_2=
\begin{bmatrix}
\frac{1}{\sqrt{2}} & 0 \\
\frac{1}{\sqrt{2}} & 0 \\
0&\frac{1}{\sqrt{2}}\\
0&\frac{1}{\sqrt{2}}
\end{bmatrix}
\begin{bmatrix}
\frac{3}{\sqrt{10}} & 0\\
0&\frac{1}{\sqrt{10}}
\end{bmatrix}
\begin{bmatrix}
\frac{1}{2}
&
\frac{1}{2}
&
\frac{1}{2}
&
\frac{1}{2}
\\
\frac{1}{2}
&
\frac{1}{2}
&
-\frac{1}{2}
&
-\frac{1}{2}
\end{bmatrix},
\:
Y_2=
\frac{3}{\sqrt{10}} a_1 b_1^*+
\frac{1}{\sqrt{10}} a_2 b_2^*,
\end{eqnarray}
where $\{a_1,a_2\}$ and $\{b_1,b_2\}$ are orthonormal sets.
We use Theorem~\ref{2.4} and
Corollary~\ref{2.3} to check whether there is a desired TPCP map for these states.
As required by condition (e) of Theorem~\ref{2.4}, we verify that the singular values of $X_1$ and $Y_1$ are related by $\{\gamma_{11},\gamma_{12}\}=\sqrt{\frac{20}{85}}\{2,1/2\}$, and those of $X_2$ and $Y_2$ are related by $\gamma_{21}=1$.
Next, we check whether there is a
tensor TPCP map sending rank one $X_{iuv}$ to rank one $Y_{iv}$.
For the left singular vectors, we seek a transformation that maps
\begin{eqnarray}
\begin{bmatrix}
1\\0\\0\\0
\end{bmatrix}
\rightarrow
a_1,
\:
\begin{bmatrix}
0\\1\\0\\0
\end{bmatrix}
\rightarrow
a_1,
\begin{bmatrix}
0\\0\\1\\0
\end{bmatrix}
\rightarrow
a_2,
\:
\begin{bmatrix}
0\\0\\0\\1
\end{bmatrix}
\rightarrow
a_2,
\:
\begin{bmatrix}
\frac{1}{\sqrt{2}}\\
\frac{1}{\sqrt{2}}\\
0\\
0
\end{bmatrix}
\rightarrow
a_1,
\:
\begin{bmatrix}
0\\
0\\
\frac{1}{\sqrt{2}}\\
\frac{1}{\sqrt{2}}
\end{bmatrix}
\rightarrow
a_2.
\end{eqnarray}

Using previous results~\cite{Uhlmann1985,Chefles200014,PhysRevA.65.052314,Chefles_Jozsa_Winter_2004},
we form the Gram matrices of the input and output states
\begin{eqnarray}
G_X=
\begin{bmatrix}
1&0&0&0&\frac{1}{\sqrt{2}}&0
\\
0&1&0&0&\frac{1}{\sqrt{2}}&0\\
0&0&1&0&0&\frac{1}{\sqrt{2}}\\
0&0&0&1&0&\frac{1}{\sqrt{2}}\\
\frac{1}{\sqrt{2}}&\frac{1}{\sqrt{2}}&0&0&1&0\\
0&0&\frac{1}{\sqrt{2}}&\frac{1}{\sqrt{2}}&0&1
\end{bmatrix},
\:
G_Y=
\begin{bmatrix}
1&1&0&0&1&0\\
1&1&0&0&1&0\\
0&0&1&1&0&1\\
0&0&1&1&0&1\\
1&1&0&0&1&0\\
0&0&1&1&0&1
\end{bmatrix},
\end{eqnarray}
and check whether there is a
correlation matrix $M$ that satisfies $G_X=M \circ G_Y$.
Clearly, $M=G_X$ is a valid choice.

For the right singular vectors, we seek a transformation that maps
\begin{eqnarray}
\begin{bmatrix}
1\\0\\0\\0
\end{bmatrix}
\rightarrow
b_1,
\:
\begin{bmatrix}
0\\1\\0\\0
\end{bmatrix}
\rightarrow
b_1,
\begin{bmatrix}
0\\0\\1\\0
\end{bmatrix}
\rightarrow
b_2,
\:
\begin{bmatrix}
0\\0\\0\\1
\end{bmatrix}
\rightarrow
b_2,
\:
\frac{1}{2}
\begin{bmatrix}
1\\
1\\
1\\
1
\end{bmatrix}
\rightarrow
b_1,
\:
\frac{1}{2}
\begin{bmatrix}
1\\
1\\
-1\\
-1
\end{bmatrix}
\rightarrow
b_2.
\end{eqnarray}
The Gram matrices are
\begin{eqnarray}
G_X=
\begin{bmatrix}
1&0&0&0&\frac{1}{2}&\frac{1}{2}
\\
0&1&0&0&\frac{1}{2}&\frac{1}{2}\\
0&0&1&0&\frac{1}{2}&-\frac{1}{2}\\
0&0&0&1&\frac{1}{2}&-\frac{1}{2}\\
\frac{1}{2}&\frac{1}{2}&\frac{1}{2}&\frac{1}{2}&1&0\\
\frac{1}{2}&\frac{1}{2}&-\frac{1}{2}&-\frac{1}{2}&0&1
\end{bmatrix},
\:
G_Y=
\begin{bmatrix}
1&1&0&0&1&0\\
1&1&0&0&1&0\\
0&0&1&1&0&1\\
0&0&1&1&0&1\\
1&1&0&0&1&0\\
0&0&1&1&0&1
\end{bmatrix}.
\end{eqnarray}
Clearly, no $M$ exists
that satisfies $G_X=M \circ G_Y$.
Therefore, no tensor TPCP map exists that transforms $X_1$ to $Y_1$ and $X_2$ to $Y_2$.

}
\end{example}


\section{Connection with other transformations}

The transformation considered in this paper can be viewed as a combination of the one-party multiple-state transformation~\cite{Uhlmann1985,Chefles200014,PhysRevA.65.052314,Chefles_Jozsa_Winter_2004} and the two-party single-state transformation~\cite{PhysRevLett.83.436}.
Let's consider an example of transforming $\{\ket{x_1}_{AB},\ket{x_2}_{AB}\}$ to $\{\ket{y_1}_{AB},\ket{y_2}_{AB}\}$ where a single channel acting on $AB$ exists to transform the two inputs to the two outputs simultaneously and two bipartite transformations exist to separately transform $\ket{x_i}_{AB}$ to $\ket{y_i}_{AB}$.
However, no bipartite LO transformation exists mapping $\ket{x_i}_{AB}$ to $\ket{y_i}_{AB}$ for $i=1,2$ simultaneously.
We use Theorem~\ref{2.1} to check.

Consider systems $A$ and $B$ where each has dimension 4.
Let
\begin{eqnarray}
\ket{x_1}_{AB}
&=&
(1.6 \ket{00} + 1.2 \ket{11} + 0.8 \ket{22} + 0.6 \ket{33})/\sqrt{5}
\\
\ket{y_1}_{AB}
&=&
0.8 \ket{00} + 0.6 \ket{11}
\\
\ket{x_2}_{AB}
&=&
(1.2 \ket{00} - 1.6 \ket{11} - 0.6 \ket{22} + 0.8 \ket{33})/\sqrt{5}
\\
\ket{y_2}_{AB}
&=&
(2 \ket{00} +  \ket{11})/\sqrt{5} .
\end{eqnarray}
Using the pure-state single-party result~\cite{Uhlmann1985,Chefles200014,PhysRevA.65.052314,Chefles_Jozsa_Winter_2004},
we can verify that a TPCP map $T$ acting on $AB$ exists such that
$T(\ket{x_i}\bra{x_i})=\ket{y_i}\bra{y_i}$
for $i = 1,  2$.
We need to check the existence of a 
correlation matrix $M$ such that
$G_X = M \circ G_Y$ where the Gram matrices $G_X=(\langle x_i \ket{x_j})=
\begin{bmatrix}
1 & 0 \\
0 & 1
\end{bmatrix}
$
and $G_y=(\langle y_i \ket{y_j})=
\begin{bmatrix}
1 & 0.98 \\
0.98 & 1
\end{bmatrix}
$.
Clearly, we can take $M=G_X$.

Next, we consider whether each input state $\ket{x_i}_{AB}$ can individually be transformed to $\ket{y_i}_{AB}$ using a bipartite LO transformation.
According to Proposition~\ref{2.2}, $X_i$ can be mapped to $Y_i$ using a bipartite LO transformation if and only if $X_i$ and $R_i \otimes Y_i$ have the same singular values, in other words,
$U_i (E_{11} \otimes X_i) V_i = R_i \otimes Y_i$.
We can verify that this is satisfied for $i=1,2$ with
\begin{eqnarray}
&
X_1=
\frac{1}{\sqrt{5}}
\begin{bmatrix}
1.6 & 0 & 0 & 0 \\
0 & 1.2 & 0 & 0 \\
0 & 0 & 0.8 & 0 \\
0 & 0 & 0 & 0.6
\end{bmatrix},
\:
R_1=
\frac{1}{\sqrt{5}}
\begin{bmatrix}
2 & 0\\
0 & 1
\end{bmatrix},
\:
Y_1=
\begin{bmatrix}
0.8 & 0 & 0 & 0\\
0 & 0.6 & 0 & 0\\
0 & 0 & 0 & 0 \\
0 & 0 & 0 & 0
\end{bmatrix},
\\
&
U_1=
\begin{bmatrix}
1 & 0 & 0 & 0 & 0 & 0 & 0 & 0\\
0 & 1 & 0 & 0 & 0 & 0 & 0 & 0\\
0 & 0 & 0 & 0 & * & * & * & *\\
0 & 0 & 0 & 0 & * & * & * & *\\
0 & 0 & 1 & 0 & 0 & 0 & 0 & 0\\
0 & 0 & 0 & 1 & 0 & 0 & 0 & 0\\
0 & 0 & 0 & 0 & * & * & * & *\\
0 & 0 & 0 & 0 & * & * & * & *
\end{bmatrix},
\:
V_1=
\begin{bmatrix}
1 & 0 & 0 & 0 & 0 & 0 & 0 & 0\\
0 & 1 & 0 & 0 & 0 & 0 & 0 & 0\\
0 & 0 & 0 & 0 & 1 & 0 & 0 & 0\\
0 & 0 & 0 & 0 & 0 & 1 & 0 & 0\\
0 & 0 & * & * & 0 & 0 & * & *\\
0 & 0 & * & * & 0 & 0 & * & *\\
0 & 0 & * & * & 0 & 0 & * & *\\
0 & 0 & * & * & 0 & 0 & * & *
\end{bmatrix},
\end{eqnarray}
and
\begin{eqnarray}
&
X_2=
\frac{1}{\sqrt{5}}
\begin{bmatrix}
1.2 & 0 & 0 & 0 \\
0 & -1.6 & 0 & 0 \\
0 & 0 & -0.6 & 0 \\
0 & 0 & 0 & 0.8
\end{bmatrix},
\:
R_2=
\begin{bmatrix}
0.6 & 0\\
0 & 0.8
\end{bmatrix},
\:
Y_2=
\frac{1}{\sqrt{5}}
\begin{bmatrix}
2 & 0 & 0 & 0 \\
0 & 1 & 0 & 0 \\
0 & 0 & 0 & 0 \\
0 & 0 & 0 & 0
\end{bmatrix},
\\
&
U_2=
\begin{bmatrix}
1 & 0 & 0 & 0 & 0 & 0 & 0 & 0\\
0 & 0 & -1 & 0 & 0 & 0 & 0 & 0\\
0 & 0 & 0 & 0 & * & * & * & *\\
0 & 0 & 0 & 0 & * & * & * & *\\
0 & -1 & 0 & 0 & 0 & 0 & 0 & 0\\
0 & 0 & 0 & 1 & 0 & 0 & 0 & 0\\
0 & 0 & 0 & 0 & * & * & * & *\\
0 & 0 & 0 & 0 & * & * & * & *
\end{bmatrix},
\:
V_2=
\begin{bmatrix}
1 & 0 & 0 & 0 & 0 & 0 & 0 & 0\\
0 & 0 & 0 & 0 & 1 & 0 & 0 & 0\\
0 & 1 & 0 & 0 & 0 & 0 & 0 & 0\\
0 & 0 & 0 & 0 & 0 & 1 & 0 & 0\\
0 & 0 & * & * & 0 & 0 & * & *\\
0 & 0 & * & * & 0 & 0 & * & *\\
0 & 0 & * & * & 0 & 0 & * & *\\
0 & 0 & * & * & 0 & 0 & * & *
\end{bmatrix}.
\end{eqnarray}
Thus, individual transformations are possible.
This is consistent with
Nielsen's result~\cite{PhysRevLett.83.436}.
Since
the eigenvalues of
$\tr_B \ket{x_i}_{AB}\bra{x_i}=X_i {X_i}^*$
are majorized by
the eigenvalues of
$\tr_B \ket{y_i}_{AB}\bra{y_i}=Y_i {Y_i}^*$,
Nielsen's result implies that $\ket{x_i}_{AB}$ can be transformed to $\ket{y_i}_{AB}$ individually using a two-way LOCC protocol (which can be a LO transformation as a special case).

On the other hand, it is obvious that no bipartite LO transformation exists to simultaneously map
$\ket{x_i}_{AB}$ to $\ket{y_i}_{AB}$ for $i=1,2$, since according to
Theorem~\ref{2.1}, we must have $U_1=U_2$ and $V_1=V_2$ for such a transformation to exist.
It is easily seen that this is not possible even if we extend the dimensions of $U_i$, $V_i$, and $R_i$.

\section{Additional Remarks}

Even in the classical case, finding a correlation matrix $M$ such that
$(x_i^*x_j) = M \circ (y_i^*y_j)$ is highly non-trivial.
If $(y_i^*y_j)$ has no zero entries, then $M$ is uniquely determined.
Clearly, if $M$ exists, then  $x_i^* x_j = 0$ whenever $y_i^*y_j = 0$.
However, in these zero positions, it is not easy to decide how to choose the corresponding
$(i,j)$ entry in $M$ so that $M$ is positive semi-definite. Such a problem is known as a completion problem in
matrix theory, and is very challenging unless the specified entries of $M$ have
some nice pattern (such as the chordal graph pattern~\cite{Laurent2001}). Nonetheless, one can use positive
semi-definite programming software to search for a solution for a given partial matrix
(i.e., a matrix for which only some of the entries are specified).

For our problem, one may search for unitary $U \in M_{mnpq}$ of the form
$U_1\otimes U_2$ with $U_1 \in M_{mp}$ and $U_2 \in M_{nq}$ such that
$$U \begin{bmatrix}\tilde x_1 & \cdots & \tilde x_k\cr 0  & \dots & 0\cr\end{bmatrix}
= 
\begin{bmatrix} \tilde{y}_1 & \cdots & \tilde{y}_k \end{bmatrix}
$$
as done in the proof of Theorem \ref{2.1}.
Alternatively, one may ignore the structure of $U$ and search for unit vectors
$\xi_1, \dots, \xi_k$ 
such that
$$(x_i^*x_j) = (\xi_i^*\xi_j) \otimes (y_i^*y_j)$$
and
$$U(E_{11} \otimes X_i)V = R_i \otimes Y_i, \qquad i = 1, \dots, k$$
as in Theorem \ref{2.1}(c).
Of course, finding $\xi_1, \dots, \xi_k$ in this way is not easy.


\section{Conclusions}

We proved necessary and sufficient conditions for the existence of a LO transformation between two sets of entangled states.
We also reduced the general problem for checking the existence of such LO transformation to smaller problems.
However, a general algorithm seems to be difficult to obtain.
On the other hand, if a LO transformation is not possible, one can easily use our theorems to rule out its existence in a computationally efficient manner.
In fact, from our theorems, one can see that it is a rather stringent condition for a LO transformation to exist, especially when the number of states to transform is large.
We hope that our results can shed some light on entanglement transformation and stimulate further investigation.
Open problems include (i) the transformability between sets of states using LOCC and separable operations and between sets of mixed states, (ii) probabilistic transformations where the final states are produced only probabilistically, and
(iii) allowing the initial states and final states to have different dimensions.

\section*{Acknowledgments}%
We thank Hoi-Kwong Lo for enlightening discussion, especially on prior works and interpretation of our result.

This research evolved in a faculty seminar on quantum information science
at the University of Hong Kong in the spring of 2012 co-ordinated by
Chau and Li. The support of the Departments of Physics and Mathematics 
of the University of Hong Kong is greatly appreciated.

Chau and Fung were partially supported by the Hong Kong RGC grant
No. 700709P. Li and Sze were partially supported by the Hong  Kong
RGC grant PolyU 502910; this grant supported the visit of Poon in
the spring of 2012.
Li was also supported by a USA NSF
grant;  he was a visiting professor of the  University of Hong Kong in the
spring of 2012,  an honorary professor of Taiyuan University of Technology
(100 Talent Program scholar), and an honorary  professor of  Shanghai
University.

\bibliographystyle{aipnum4-1}

\bibliography{qc57.3}

\begin{thebibliography}{24}%
\makeatletter
\providecommand \@ifxundefined [1]{%
 \@ifx{#1\undefined}
}%
\providecommand \@ifnum [1]{%
 \ifnum #1\expandafter \@firstoftwo
 \else \expandafter \@secondoftwo
 \fi
}%
\providecommand \@ifx [1]{%
 \ifx #1\expandafter \@firstoftwo
 \else \expandafter \@secondoftwo
 \fi
}%
\providecommand \natexlab [1]{#1}%
\providecommand \enquote  [1]{``#1''}%
\providecommand \bibnamefont  [1]{#1}%
\providecommand \bibfnamefont [1]{#1}%
\providecommand \citenamefont [1]{#1}%
\providecommand \href@noop [0]{\@secondoftwo}%
\providecommand \href [0]{\begingroup \@sanitize@url \@href}%
\providecommand \@href[1]{\@@startlink{#1}\@@href}%
\providecommand \@@href[1]{\endgroup#1\@@endlink}%
\providecommand \@sanitize@url [0]{\catcode `\\12\catcode `\$12\catcode
  `\&12\catcode `\#12\catcode `\^12\catcode `\_12\catcode `\%12\relax}%
\providecommand \@@startlink[1]{}%
\providecommand \@@endlink[0]{}%
\providecommand \url  [0]{\begingroup\@sanitize@url \@url }%
\providecommand \@url [1]{\endgroup\@href {#1}{\urlprefix }}%
\providecommand \urlprefix  [0]{URL }%
\providecommand \Eprint [0]{\href }%
\providecommand \doibase [0]{http://dx.doi.org/}%
\providecommand \selectlanguage [0]{\@gobble}%
\providecommand \bibinfo  [0]{\@secondoftwo}%
\providecommand \bibfield  [0]{\@secondoftwo}%
\providecommand \translation [1]{[#1]}%
\providecommand \BibitemOpen [0]{}%
\providecommand \bibitemStop [0]{}%
\providecommand \bibitemNoStop [0]{.\EOS\space}%
\providecommand \EOS [0]{\spacefactor3000\relax}%
\providecommand \BibitemShut  [1]{\csname bibitem#1\endcsname}%
\let\auto@bib@innerbib\@empty
\bibitem [{\citenamefont {Ekert}(1991)}]{PhysRevLett.67.661}%
  \BibitemOpen
  \bibfield  {author} {\bibinfo {author} {\bibfnamefont {A.~K.}\ \bibnamefont
  {Ekert}},\ }\href {\doibase 10.1103/PhysRevLett.67.661} {\bibfield  {journal}
  {\bibinfo  {journal} {Phys. Rev. Lett.}\ }\textbf {\bibinfo {volume} {67}},\
  \bibinfo {pages} {661} (\bibinfo {year} {1991})}\BibitemShut {NoStop}%
\bibitem [{\citenamefont {Bennett}, \citenamefont {Brassard},\ and\
  \citenamefont {Mermin}(1992)}]{PhysRevLett.68.557}%
  \BibitemOpen
  \bibfield  {author} {\bibinfo {author} {\bibfnamefont {C.~H.}\ \bibnamefont
  {Bennett}}, \bibinfo {author} {\bibfnamefont {G.}~\bibnamefont {Brassard}}, \
  and\ \bibinfo {author} {\bibfnamefont {N.~D.}\ \bibnamefont {Mermin}},\
  }\href {\doibase 10.1103/PhysRevLett.68.557} {\bibfield  {journal} {\bibinfo
  {journal} {Phys. Rev. Lett.}\ }\textbf {\bibinfo {volume} {68}},\ \bibinfo
  {pages} {557} (\bibinfo {year} {1992})}\BibitemShut {NoStop}%
\bibitem [{\citenamefont {Lo}\ and\ \citenamefont {Chau}(1999)}]{LoChauQKD_99}%
  \BibitemOpen
  \bibfield  {author} {\bibinfo {author} {\bibfnamefont {H.-K.}\ \bibnamefont
  {Lo}}\ and\ \bibinfo {author} {\bibfnamefont {H.~F.}\ \bibnamefont {Chau}},\
  }\href@noop {} {\bibfield  {journal} {\bibinfo  {journal} {Science}\ }\textbf
  {\bibinfo {volume} {283}},\ \bibinfo {pages} {2050} (\bibinfo {year}
  {1999})}\BibitemShut {NoStop}%
\bibitem [{\citenamefont {Deutsch}(1985)}]{Deutsch1985}%
  \BibitemOpen
  \bibfield  {author} {\bibinfo {author} {\bibfnamefont {D.}~\bibnamefont
  {Deutsch}},\ }\href@noop {} {\bibfield  {journal} {\bibinfo  {journal} {Proc.
  Royal Soc. London Series A}\ }\textbf {\bibinfo {volume} {400}},\ \bibinfo
  {pages} {97} (\bibinfo {year} {1985})}\BibitemShut {NoStop}%
\bibitem [{\citenamefont {Shor}(1994)}]{Shor1994}%
  \BibitemOpen
  \bibfield  {author} {\bibinfo {author} {\bibfnamefont {P.}~\bibnamefont
  {Shor}},\ }in\ \href {\doibase 10.1109/SFCS.1994.365700} {\emph {\bibinfo
  {booktitle} {Foundations of Computer Science, 1994 Proceedings., 35th Annual
  Symposium on}}}\ (\bibinfo {year} {1994})\ pp.\ \bibinfo {pages}
  {124--134}\BibitemShut {NoStop}%
\bibitem [{\citenamefont {DiVincenzo}(1995)}]{DiVincenzo1995}%
  \BibitemOpen
  \bibfield  {author} {\bibinfo {author} {\bibfnamefont {D.~P.}\ \bibnamefont
  {DiVincenzo}},\ }\href@noop {} {\bibfield  {journal} {\bibinfo  {journal}
  {Science}\ }\textbf {\bibinfo {volume} {270}},\ \bibinfo {pages} {255}
  (\bibinfo {year} {1995})}\BibitemShut {NoStop}%
\bibitem [{\citenamefont {Lloyd}(1996)}]{Lloyd1996}%
  \BibitemOpen
  \bibfield  {author} {\bibinfo {author} {\bibfnamefont {S.}~\bibnamefont
  {Lloyd}},\ }\href@noop {} {\bibfield  {journal} {\bibinfo  {journal}
  {Science}\ }\textbf {\bibinfo {volume} {273}},\ \bibinfo {pages} {1073}
  (\bibinfo {year} {1996})}\BibitemShut {NoStop}%
\bibitem [{\citenamefont {Shor}(1996)}]{Shor1996}%
  \BibitemOpen
  \bibfield  {author} {\bibinfo {author} {\bibfnamefont {P.}~\bibnamefont
  {Shor}},\ }in\ \href {\doibase 10.1109/SFCS.1996.548464} {\emph {\bibinfo
  {booktitle} {Foundations of Computer Science, 1996. Proceedings., 37th Annual
  Symposium on}}}\ (\bibinfo {year} {1996})\ pp.\ \bibinfo {pages}
  {56--65}\BibitemShut {NoStop}%
\bibitem [{\citenamefont {Cirac}, \citenamefont {Pellizzari},\ and\
  \citenamefont {Zoller}(1996)}]{Cirac1996}%
  \BibitemOpen
  \bibfield  {author} {\bibinfo {author} {\bibfnamefont {J.~I.}\ \bibnamefont
  {Cirac}}, \bibinfo {author} {\bibfnamefont {T.}~\bibnamefont {Pellizzari}}, \
  and\ \bibinfo {author} {\bibfnamefont {P.}~\bibnamefont {Zoller}},\
  }\href@noop {} {\bibfield  {journal} {\bibinfo  {journal} {Science}\ }\textbf
  {\bibinfo {volume} {273}},\ \bibinfo {pages} {1207} (\bibinfo {year}
  {1996})}\BibitemShut {NoStop}%
\bibitem [{\citenamefont {Gottesman}(1998)}]{PhysRevA.57.127}%
  \BibitemOpen
  \bibfield  {author} {\bibinfo {author} {\bibfnamefont {D.}~\bibnamefont
  {Gottesman}},\ }\href {\doibase 10.1103/PhysRevA.57.127} {\bibfield
  {journal} {\bibinfo  {journal} {Phys. Rev. A}\ }\textbf {\bibinfo {volume}
  {57}},\ \bibinfo {pages} {127} (\bibinfo {year} {1998})}\BibitemShut
  {NoStop}%
\bibitem [{\citenamefont {Alberti}\ and\ \citenamefont
  {Uhlmann}(1980)}]{Alberti1980163}%
  \BibitemOpen
  \bibfield  {author} {\bibinfo {author} {\bibfnamefont {P.}~\bibnamefont
  {Alberti}}\ and\ \bibinfo {author} {\bibfnamefont {A.}~\bibnamefont
  {Uhlmann}},\ }\href {\doibase 10.1016/0034-4877(80)90083-X} {\bibfield
  {journal} {\bibinfo  {journal} {Rep. Math. Phys.}\ }\textbf {\bibinfo
  {volume} {18}},\ \bibinfo {pages} {163 } (\bibinfo {year}
  {1980})}\BibitemShut {NoStop}%
\bibitem [{\citenamefont {Uhlmann}(1985)}]{Uhlmann1985}%
  \BibitemOpen
  \bibfield  {author} {\bibinfo {author} {\bibfnamefont {A.}~\bibnamefont
  {Uhlmann}},\ }\href@noop {} {\bibfield  {journal} {\bibinfo  {journal} {Wiss.
  Z. Karl-Marx-Univ. Leipzig, Math.-Naturwiss. Reihe}\ }\textbf {\bibinfo
  {volume} {34}},\ \bibinfo {pages} {580} (\bibinfo {year} {1985})}\BibitemShut
  {NoStop}%
\bibitem [{\citenamefont {Chefles}(2000)}]{Chefles200014}%
  \BibitemOpen
  \bibfield  {author} {\bibinfo {author} {\bibfnamefont {A.}~\bibnamefont
  {Chefles}},\ }\href {\doibase 10.1016/S0375-9601(00)00291-7} {\bibfield
  {journal} {\bibinfo  {journal} {Phys. Lett. A}\ }\textbf {\bibinfo {volume}
  {270}},\ \bibinfo {pages} {14 } (\bibinfo {year} {2000})}\BibitemShut
  {NoStop}%
\bibitem [{\citenamefont {Chefles}(2002)}]{PhysRevA.65.052314}%
  \BibitemOpen
  \bibfield  {author} {\bibinfo {author} {\bibfnamefont {A.}~\bibnamefont
  {Chefles}},\ }\href {\doibase 10.1103/PhysRevA.65.052314} {\bibfield
  {journal} {\bibinfo  {journal} {Phys. Rev. A}\ }\textbf {\bibinfo {volume}
  {65}},\ \bibinfo {pages} {052314} (\bibinfo {year} {2002})}\BibitemShut
  {NoStop}%
\bibitem [{\citenamefont {Chefles}, \citenamefont {Jozsa},\ and\ \citenamefont
  {Winter}(2004)}]{Chefles_Jozsa_Winter_2004}%
  \BibitemOpen
  \bibfield  {author} {\bibinfo {author} {\bibfnamefont {A.}~\bibnamefont
  {Chefles}}, \bibinfo {author} {\bibfnamefont {R.}~\bibnamefont {Jozsa}}, \
  and\ \bibinfo {author} {\bibfnamefont {A.}~\bibnamefont {Winter}},\
  }\href@noop {} {\bibfield  {journal} {\bibinfo  {journal} {Int. J. Quant.
  Inf.}\ }\textbf {\bibinfo {volume} {2}},\ \bibinfo {pages} {11} (\bibinfo
  {year} {2004})}\BibitemShut {NoStop}%
\bibitem [{\citenamefont {Bennett}\ \emph {et~al.}(1996)\citenamefont
  {Bennett}, \citenamefont {Bernstein}, \citenamefont {Popescu},\ and\
  \citenamefont {Schumacher}}]{Bennett1996}%
  \BibitemOpen
  \bibfield  {author} {\bibinfo {author} {\bibfnamefont {C.~H.}\ \bibnamefont
  {Bennett}}, \bibinfo {author} {\bibfnamefont {H.~J.}\ \bibnamefont
  {Bernstein}}, \bibinfo {author} {\bibfnamefont {S.}~\bibnamefont {Popescu}},
  \ and\ \bibinfo {author} {\bibfnamefont {B.}~\bibnamefont {Schumacher}},\
  }\href {\doibase 10.1103/PhysRevA.53.2046} {\bibfield  {journal} {\bibinfo
  {journal} {Phys. Rev. A}\ }\textbf {\bibinfo {volume} {53}},\ \bibinfo
  {pages} {2046} (\bibinfo {year} {1996})}\BibitemShut {NoStop}%
\bibitem [{\citenamefont {Lo}\ and\ \citenamefont
  {Popescu}(2001)}]{Lo:2001:concentration}%
  \BibitemOpen
  \bibfield  {author} {\bibinfo {author} {\bibfnamefont {H.-K.}\ \bibnamefont
  {Lo}}\ and\ \bibinfo {author} {\bibfnamefont {S.}~\bibnamefont {Popescu}},\
  }\href {\doibase 10.1103/PhysRevA.63.022301} {\bibfield  {journal} {\bibinfo
  {journal} {Phys. Rev. A}\ }\textbf {\bibinfo {volume} {63}},\ \bibinfo
  {pages} {022301} (\bibinfo {year} {2001})}\BibitemShut {NoStop}%
\bibitem [{\citenamefont {Nielsen}(1999)}]{PhysRevLett.83.436}%
  \BibitemOpen
  \bibfield  {author} {\bibinfo {author} {\bibfnamefont {M.~A.}\ \bibnamefont
  {Nielsen}},\ }\href {\doibase 10.1103/PhysRevLett.83.436} {\bibfield
  {journal} {\bibinfo  {journal} {Phys. Rev. Lett.}\ }\textbf {\bibinfo
  {volume} {83}},\ \bibinfo {pages} {436} (\bibinfo {year} {1999})}\BibitemShut
  {NoStop}%
\bibitem [{\citenamefont {Jonathan}\ and\ \citenamefont
  {Plenio}(1999)}]{PhysRevLett.83.1455}%
  \BibitemOpen
  \bibfield  {author} {\bibinfo {author} {\bibfnamefont {D.}~\bibnamefont
  {Jonathan}}\ and\ \bibinfo {author} {\bibfnamefont {M.~B.}\ \bibnamefont
  {Plenio}},\ }\href {\doibase 10.1103/PhysRevLett.83.1455} {\bibfield
  {journal} {\bibinfo  {journal} {Phys. Rev. Lett.}\ }\textbf {\bibinfo
  {volume} {83}},\ \bibinfo {pages} {1455} (\bibinfo {year}
  {1999})}\BibitemShut {NoStop}%
\bibitem [{\citenamefont {He}\ and\ \citenamefont {Bergou}(2008)}]{He2008}%
  \BibitemOpen
  \bibfield  {author} {\bibinfo {author} {\bibfnamefont {B.}~\bibnamefont
  {He}}\ and\ \bibinfo {author} {\bibfnamefont {J.~A.}\ \bibnamefont
  {Bergou}},\ }\href {\doibase 10.1103/PhysRevA.78.062328} {\bibfield
  {journal} {\bibinfo  {journal} {Phys. Rev. A}\ }\textbf {\bibinfo {volume}
  {78}},\ \bibinfo {pages} {062328} (\bibinfo {year} {2008})}\BibitemShut
  {NoStop}%
\bibitem [{\citenamefont {Gheorghiu}\ and\ \citenamefont
  {Griffiths}(2008)}]{Gheorghiu2008}%
  \BibitemOpen
  \bibfield  {author} {\bibinfo {author} {\bibfnamefont {V.}~\bibnamefont
  {Gheorghiu}}\ and\ \bibinfo {author} {\bibfnamefont {R.~B.}\ \bibnamefont
  {Griffiths}},\ }\href {\doibase 10.1103/PhysRevA.78.020304} {\bibfield
  {journal} {\bibinfo  {journal} {Phys. Rev. A}\ }\textbf {\bibinfo {volume}
  {78}},\ \bibinfo {pages} {020304} (\bibinfo {year} {2008})}\BibitemShut
  {NoStop}%
\bibitem [{\citenamefont {Huang}\ \emph {et~al.}(2012)\citenamefont {Huang},
  \citenamefont {Li}, \citenamefont {Poon},\ and\ \citenamefont
  {Sze}}]{Huang2012}%
  \BibitemOpen
  \bibfield  {author} {\bibinfo {author} {\bibfnamefont {Z.}~\bibnamefont
  {Huang}}, \bibinfo {author} {\bibfnamefont {C.-K.}\ \bibnamefont {Li}},
  \bibinfo {author} {\bibfnamefont {E.}~\bibnamefont {Poon}}, \ and\ \bibinfo
  {author} {\bibfnamefont {N.-S.}\ \bibnamefont {Sze}},\ }\href@noop {} {\
  (\bibinfo {year} {2012})},\ \bibinfo {note} {to appear in J. Math. Phys.},\
  \Eprint {http://arxiv.org/abs/{e-print} arXiv:1203.5547 [math-ph]} {{e-print}
  arXiv:1203.5547 [math-ph]} \BibitemShut {NoStop}%
\bibitem [{\citenamefont {Horn}\ and\ \citenamefont
  {Johnson}(1994)}]{Horn1994}%
  \BibitemOpen
  \bibfield  {author} {\bibinfo {author} {\bibfnamefont {R.~A.}\ \bibnamefont
  {Horn}}\ and\ \bibinfo {author} {\bibfnamefont {C.~R.}\ \bibnamefont
  {Johnson}},\ }\href@noop {} {\emph {\bibinfo {title} {Topics in Matrix
  Analysis}}}\ (\bibinfo  {publisher} {Cambridge University Press},\ \bibinfo
  {year} {1994})\ Chap.~\bibinfo {chapter} {4}\BibitemShut {NoStop}%
\bibitem [{\citenamefont {Laurent}(2001)}]{Laurent2001}%
  \BibitemOpen
  \bibfield  {author} {\bibinfo {author} {\bibfnamefont {M.}~\bibnamefont
  {Laurent}},\ }\enquote {\bibinfo {title} {Matrix completion problems},}\ in\
  \href@noop {} {\emph {\bibinfo {booktitle} {The Encyclopedia of Optimization,
  vol. III (Interior - M)}}},\ \bibinfo {editor} {edited by\ \bibinfo {editor}
  {\bibfnamefont {C.}~\bibnamefont {Floudas}}\ and\ \bibinfo {editor}
  {\bibfnamefont {P.}~\bibnamefont {Pardalos}}}\ (\bibinfo  {publisher}
  {Kluwer},\ \bibinfo {year} {2001})\ pp.\ \bibinfo {pages}
  {221--229}\BibitemShut {NoStop}%
\end{thebibliography}%

\end{document}